\documentclass{article}

\usepackage{amsfonts,amssymb,amsmath,amsthm}
\usepackage{shuffle}
\usepackage[T2A]{fontenc}
\usepackage[utf8]{inputenc}        

\usepackage{cite} 

\usepackage{microtype}

\usepackage[makeroom]{cancel} 
\usepackage{url}
\let\ov\overline

\def\getbi{\mathsf{get}_\mathsf{B_I}}
\def\updbs{\mathsf{upd}_\mathsf{B_S}}
\def\wr{\textsf{write}}
\def\query{\mathsf{query}}
\def\q{\mathsf{q}}
\def\resp{\mathsf{resp}}
\def\r{\mathsf{r}}
\def\valid{\mathsf{valid}}

\def\push{\mathsf{push}}
\def\pop{\mathsf{pop}}

\def\MP{\mathrm{M}\mm}
\def\EP{\mathrm{E}\mm}
\def\BA{\mathrm{BA}}
\def\NBA{\mathrm{NBA}}
\def\DBA{\mathrm{DBA}}
\def\NA{\mathrm{NA}}
\def\DA{\mathrm{DA}}

\def\BP{\mathrm{B_{\PROT}}}
\def\NBP{\mathrm{N}\BP}
\def\DBP{\mathrm{D}\BP}
\def\NBPA{\NBP\mathrm{A}}
\def\DBPA{\DBP\mathrm{A}}
\def\BPA{\BP\mathrm{A}}

\def\sq{\ensuremath{\mathrm {sq}}}

\def\test{\mathsf{test}}
  
 \def\ins{\mathsf{in}}
 \def\out{\mathsf{out}}

\def\logTM{{\log}\mm\mathrm{TM}}

	\def\ders{ \vdash^{\!\!\! {}^{ *}}}
	\def\der{ \vdash }

\let\lex\prec
\let\lexeq\preceq

\def\poly{\mathop{\mathrm{poly}}\nolimits}  

\def\leGen#1#2{\mathop{\leq^{\mathrm{#2}}_{\mathrm{#1}}}}

\def\lelog{\leGen{log}{}}

\def\lerat{\mathop{\leq_{\mathrm{rat}}}}

\def\simrat{\mathop{\sim_{\mathrm{rat}}}}

\def\leT{\mathop{\leq_{\mathrm{T}}}}
\def\leTp{\leGen{T}{P}}

\def\Rt{\mathrm{Right}}
\def\Lft{\mathrm{Left}}
\def\reg{\mathrm{DRR}}
\def\nreg{\mathrm{NRR}}
\def\run#1{\xhookrightarrow{#1}}
\def\Per{\mathrm{Per}}

\let\es\varnothing

\def\Pref{\mathrm{Pref}}

\def\CC{\mathbb C}

\def\SS{\mathbb S}
\def\ZZ{\mathbb Z}

\let\epsilon\varepsilon
\let\eps\varepsilon

\let\rendmarker\vartriangleleft

\def\A{{\cal A}}
\def\B{{\cal B}}
\def\CC{{\cal C}}

\def\T{{\cal T}}

\def\0{\mathsf{false}}
\def\1{\mathsf{true}}

\def\PP{{\mathbf{P}}}
\def\NP{{\mathbf{NP}}}

\def\CFL{{\mathsf{CFL}}}

\def\SAPROT{{\text{\sf{SA\text{-}PROT}}}}
\def\SPkPROT{{\mathsf{S_{1,k}PROT}}}

\def\DPROT{{\mathsf{D_2}\text{-}\mathsf{PROT}}}
\def\PROT{\mathsf{P}}

\def\PSPACE{{\mathbf{PSPACE}}}
\def\CFL{{\mathsf{CFL}}}
\def\C{\mathbf{C}}

\usepackage{mathrsfs}

\def\CL{{\mathscr L}}
\def\CM{{\mathscr M}}
\def\CF{{\mathscr F}}

\newcommand{\Rnum}[1]{\expandafter{\romannumeral #1\relax}}
\newcommand{\RNum}[1]{\uppercase\expandafter{\romannumeral #1\relax}}

\usepackage{enumitem}
\setlist[enumerate]{before=\setupmodenumerate}

\newif\ifmoditem
\newcommand{\setupmodenumerate}{%
  \global\moditemfalse
  \let\origmakelabel\makelabel
  \def\moditem##1{\global\moditemtrue\def\mesymbol{##1}\item}%
  \def\makelabel##1{%
    \origmakelabel{\ifmoditem\llap{\mesymbol\enspace}\fi##1}%
    \global\moditemfalse}%
}

	\usepackage{mathtools} 

\def\mm{\text{-}}

\makeatletter
\def\zzlabel#1{\ifmeasuring@\else\ltx@label{#1}\fi}
\makeatletter

\let\lendmarker\vartriangleright
\let\rendmarker\vartriangleleft

\def\aux{\mathsf{aux}}  
\def\mem{\mathsf{mem}}

\theoremstyle{plain}

\newtheorem{theorem}{Theorem}
\newtheorem{lemma}[theorem]{Lemma}
\newtheorem{proposition}[theorem]{Proposition}

\theoremstyle{definition}
\newtheorem{definition}[theorem]{Definition}
\newtheorem{remark}[theorem]{Remark}
\newtheorem{example}[theorem]{Example}
\let\definitionbodyfont\relax
\let\remarkbodyfont\relax
\let\examplebodyfont\relax

\title{Automata Equipped with Auxiliary Data Structures and Regular Realizability Problems}

\author{Alexander Rubtsov\thanks{Faculty of Computer Science, National Research University Higher School of Economics,
 Pokrovsky boulevard 11,   Moscow, 109028,  Russia,
 \texttt{rubtsov99@gmail.com}} \and
  Mikhail Vyalyi\thanks{Faculty of Computer Science, National Research University Higher School of Economics,
 Pokrovsky boulevard 11,   Moscow, 109028,  Russia,
 \texttt{vyalyi@gmail.com}
  }
}

\date{\today}

\begin{document}

\maketitle

\begin{abstract}
We consider general computational models: one-way and two-way finite automata, and logarithmic space Turing machines, all equipped with an auxiliary data structure (ADS). The definition of an ADS is based on the language of protocols of work with the ADS.
We describe the connection of automata-based models with ``Balloon automata'' that are another general formalization of automata equipped with an ADS
presented by Hopcroft and Ullman in 1967.
 This definition establishes the connection between the non-emptiness problem for one-way automata with ADS, languages recognizable by nondeterministic log-space Turing machines equipped with the same ADS, and a regular realizability problem (NRR) for the language of ADS' protocols. The NRR problem is to verify whether the regular language on the input has a non-empty intersection with the language of protocols. The computational complexity of these problems (and languages) is the same up to log-space reductions.

\end{abstract}

\section{Introduction}\label{Intro}

Many computational models  are derived from (one-way) finite automata (FAs) via equipping them with an auxiliary data structure (ADS). The best-known model of this kind is pushdown automata (PDAs), the deterministic version of which is widely used in compilers. Other examples are $k$-counter automata, $(k,r)$-reversal-bounded counter automata (equipped with $k$ counters each of which can switch between increasing and decreasing modes at most $r$ times), stack automata, nested stack automata, bag automata~\cite{BagAut}, set automata (SAs)~\cite{KutribSApaper2016} and their another variant~\cite{Lange96setautomata}; more examples can be found in~\cite{BalloonHU67}.

During the investigation of balloon automata (BAs)~\cite{BalloonHU67}, Hopcroft and Ullman connected the decidability of the
membership and the emptiness problems for one-way and two-way models; we denote them as $M\mm xyBA$ and $E\mm xyBA$
respectively, where $x=1$ denotes one-way and $x=2$ denotes two-way models, and $y \in \{D, N\}$ stands for determinism or nondeterminism respectively. 
Eq.~\eqref{eq:HUresults} summarizes results on decidability questions from~\cite{BalloonHU67}, where $\leT$ is a \emph{Turing-reduction} and $\{A, B\}$ means that $A \leT B$ and $B\leT A$. 
\begin{align}\label{eq:HUresults}
	\begin{split}
		\{\MP1\DBA, \MP2\DBA\} & \leT \{\EP1\DBA,\EP1\NBA, \MP1\NBA, \MP2\NBA\} \leT \\
							   & \leT \EP2\DBA \leT \EP2\NBA.
	\end{split}
\end{align}
We remark that the relation $\EP1\NBA\leT\EP1\DBA$ was proved for the case of at least a two-letter input alphabet.

While a lot of  models can be described as BA, it is hard to invent such a model with good computational properties. One of the reasons is that the equipment of finite automata 
with a complex data structure (or with several simple data structures) often leads to a universal computational model.
For example, FAs equipped with two pushdown stores are equivalent to Turing machines (TMs), as well as FAs equipped with two non-restricted counters. 

In this paper, we investigate the computational power of FAs equipped with an ADS.
We describe the model using the language of correct protocols of work with the ADS.
 We provide a general approach to analyze the complexity of the
 emptiness problem and prove the following non-trivial result. If FAs are equipped with an ADS and  nondeterministic logarithmic space TMs ($\logTM$s, see the definition
 in~\cite{Sipser}) are equipped with the same ADS, then the FAs' non-emptines problem and the TMs-recognizable languages are of the same complexity (up to log-space reductions). Our key tool is
 the regular realizability problem
 (see Definition~\ref{def:RR} below).

\subsection{Our Contribution}\label{sec:results}
BAs were initially defined as automata with access to additional 
storage of unspecified structure---\emph{the balloon}. A rather general axioms were imposed
for the balloon and the interaction of the balloon and the automaton
(see Definition~\ref{def::BA} below).
In this paper, we propose another definition based on a language of the ADS' protocols that we denote as~$\PROT$, so we refer to the ADS as $\BP$.
We prove that languages recognizable by $1\NBPA$ form not just a rational cone as in the case of $1\NBA$~\cite{BalloonHU67}, but a principal rational cone generated by $\PROT$ (we provide the definition in Section~\ref{sect:RatTransd}).

 This reformulation guarantees good structural properties, some of them follow from the connection with BA (Section~\ref{sect:BA}), and provides the relation between $\EP1\NBPA$ and the \emph{nondeterministic regular realizability problem}.

\begin{definition}\label{def:RR}\definitionbodyfont 
	Fix a formal language $F$ called a \emph{filter}, the parameter of \emph{regular realizability problems} $\reg(F)$ and $\nreg(F)$ that are the problems
	of verifying non-emptiness of the intersection of the filter $F$
with a regular language $L(\A)$ described via the DFA or NFA~$\A$  respectively.
Formally,
\[
\begin{aligned}
  &\nreg(F) = \{ \A \mid \A \text{ is an NFA and } L(\A)\cap F \neq \es  \},\\
  &\reg(F) = \{ \A \mid \A \text{ is a DFA and } L(\A)\cap F \neq \es  \}.
 \end{aligned}
\]
\end{definition}

RR problems have independently been studied under the name regular intersection emptiness problems~\cite{WolfDCFS19,WolfFernauIntregCoRR20}. A restricted version of RR problem (for context-free filters only) is a well-known CFL-reachability problem, which is related to problems in interprocedural program analysis~\cite{Chistikov22, DOLEV198257, Yannakakis90, melski2000interconvertibility, Reps95, bouajjani1997reachability}.

In this paper we focus on the computational complexity, so we use the weakest reduction suitable for our needs,
the deterministic log-space reduction that we denote as $\lelog$. If $A \lelog B$ and $B \lelog A$ we write $A \sim_{\log} B$ and say that $A$ and $B$ are \emph{log-space equivalent}. Note that in our constructions, emptiness and membership problems are the sets of instances' descriptions with positive answers, i.e., $\EP xy\BPA = \{ \langle M \rangle \mid L(M) = \es \}$, $\MP xy\BPA = \{ \langle M, w\rangle \mid w \in L(M) \}$, where $M$ is a $xy\BPA$ and $\langle x \rangle$ is the description of $x$. So, $\ov{\EP xy\BPA} = \{ \langle M\rangle \mid L(M) \neq \es \}$.  We prove that
	 $\ov{\EP1\NBPA} \sim_{\log} \nreg(\PROT)  $.
Based on this result, we establish computational universality of $\ov{\EP1\NBPA}  $ (see Theorem~\ref{th:BP-universality} below). Note that in the universality result we need Turing reductions in polynomial time instead of log-space reductions.

We equip with ADS not only FAs but also $\logTM$s. 
We denote deterministic and nondeterministic $\logTM$s equipped with an ADS~$\BP$ as $\DBP\logTM$ and $\NBP\logTM$ 
respectively. We prove that
\begin{equation}\label{eq:NRRlogTMresult}
 \nreg(\PROT) \sim_{\log} \CL(\NBP\logTM) = \{ L \mid L \lelog \nreg(\PROT) \},
\end{equation}
hereinafter $\CL(\text{model})$ is the class of languages recognizable by the model. 
If $P$ is a problem (formal language) and $S$ is a set of problems (class of formal languages) the reductions mean as follows. $ P \leq S$ means that $\exists P' \in S : P \leq P'$ 
and $ S \leq P$ means that $\forall P' \in S : P' \leq P$; $S\sim P$ means $(P \leq S)\land(S \leq P)$.

It is easy to verify that in the original proofs in~\cite{BalloonHU67}, Turing reductions in~\eqref{eq:HUresults} 
can be replaced by the log-space reductions provided we replace the emptiness problems with non-emptiness ones. So, we obtain
\begin{align}\label{eq:HUresultsLog}
	\begin{split}
		& \{\MP1\DBPA, \MP2\DBPA\} \lelog   \{\ov{\EP1\DBPA},\ov{\EP1\NBPA}, \MP1\NBPA, \MP2\NBPA, \\
		&  \nreg(\PROT), \CL(\NBP\logTM)\} \lelog 
		 \ov{\EP2\DBA} \lelog \ov{\EP2\NBA} \lelog \ov{\EP\NBP\logTM}.
	\end{split}
\end{align}
We also prove the reduction 
\begin{equation}
	\MP1\DBPA \lelog \reg(\PROT).
\end{equation}
 
Results~\eqref{eq:HUresultsLog} combined with  known facts imply assertions~(\ref{fact:P}-\ref{fact:NP}), where S is the set data structure as in SA, S$_1$ is the set data structure that supports the insertion of at most one word, that cannot be removed further but can be tested if a query-word in the set.
In S$_{1, |\Gamma| = 1}$ the word in the set is over an unary alphabet, $\PSPACE\textbf{\mm c}$ and $\NP\textbf{\mm c}$ are subclasses of complete languages.
\begin{align}
		\label{fact:P} 		&\PP = \CL(\mathrm{NPD}\logTM),\text{ where PD is Pushdown store},\\
		&\PSPACE \supseteq  \CL(\mathrm{NS}\logTM), \exists L \in \CL(\mathrm{NS}\logTM) : L \in \PSPACE\textbf{\mm c} \label{fact:PSPACEI},\\
		&\PSPACE  \supseteq \CL(\mathrm{NS}_{1}\logTM), \exists L \in \CL(\mathrm{NS}_{1}\logTM) : L \in \PSPACE\textbf{\mm c},  \label{fact:PSPACEII}\\ 
	&\NP \supseteq \CL(\mathrm{NS}_{1, |\Gamma| = 1}\logTM), , \exists L \in \CL(\mathrm{NS}_{1, |\Gamma| = 1}\logTM) : L \in \NP\textbf{\mm c}. \label{fact:NP} 	
\end{align}

Assertion~\eqref{fact:P} is a well-known fact. Our technique here just shows a new connection: \eqref{fact:P} directly follows from the fact that the emptiness problem for PDA is $\PP$-complete. Assertions (\ref{fact:PSPACEI}-\ref{fact:NP}) are new results to the best of our knowledge, we prove them in Section~\ref{sect:Applications}.
Assertions~(\ref{fact:PSPACEII}-\ref{fact:NP}) lead to~\eqref{eq:HUresultsLog} for the corresponding classes of automata.
For~\eqref{fact:PSPACEI}, we have already obtained the result  in~\cite{RVCSR2018} in the same way and present in this paper the generalized technique.

\section{Definitions}\label{sect:Defs}

\subsection{Notation on binary relations}

We associate with a binary relation $R\subseteq A\times B$ the corresponding mappings
$A \to 2^{B} $ and $2^{A} \to 2^B$ that are denoted by the same letter $R$, so $R(a) = \{ b : aRb\}$ and $R(S) = \cup_{a\in S} R(a)$. A relation $R$ is the \emph{composition} of the relations $P \subseteq A \times C$ and $Q \subseteq C \times B$ if $ R = \{ (a, b) \mid \exists c : aPc \land cQb \} $; we denote the composition as $Q\circ P$. In the case of a set $S \subseteq C$ we treat $S$ as a binary relation $S \subseteq  C\times \{0,1\}$ in the composition $S \circ P = S'$ that returns the set~$S'\subseteq A$. We denote the reflexive and transitive closure of $R \subseteq A\times A$ by $R^*$; the symbol $*$ can also be placed above the relation, e.g., $\ders$. We denote by $R^{-1} \subseteq B \times A$ the inverse relation, i.e., $aRb \iff bR^{-1}a$.

\subsection{Rational Transductions}\label{sect:RatTransd}

Our technique is based on the connection of NRR problems with rational cones. We recall the definitions borrowing them from the book~\cite{BerstelBook}.
	A \emph{finite state transducer} (FST) is a nondeterministic finite automaton with 
        an output tape, and DFST is the deterministic version of FST. For the deterministic version, it is important that a transducer can write a word (but not only a single symbol) on the output tape on processing a letter from the input tape.
        Let $T$ be an FST; we also denote by $T$ the corresponding relation, i.e., $uTv$
        if  there exists a run of $T$ on the input $u$ from the initial state to a final state such that at the end of the run the word~$v$ is written on the output tape.
The \emph{rational dominance} relation $A \lerat B$ holds
if there exists an FST $T$ such that $A = T(B)$, here $A$ and $B$ are languages. The relations computable by FSTs are known as \emph{rational relations}.
The following lemmata are algorithmic versions of well-known facts (see \cite{BerstelBook}, Chapter III), the first one is the algorithmic version of the Elgot-Mezei theorem. The log-space algorithms follow from straight-forward constructions.

\begin{lemma}\label{lemma:EM}
	For FSTs $T_1$ and $T_2$ such that $T_1 \subseteq \Sigma^*\times\Delta^*$, $T_2 \subseteq \Delta^*\times\Gamma^*$, and FA $\A$ such that $L(\A) \subseteq \Delta^*$, there exists an FST $T$ such that $T = T_2 \circ T_1 \subseteq \Sigma^*\times\Gamma^*$, and NFA $\B$
	recognizing the language $T_1^{-1}L(\A)$. So, the relation $\lerat$ is transitive. Moreover, $T$ and $\B$ are constructible in logarithmic space. We denote FST~$T$ and NFA~$\B$ as $T_2\circ T_1 $ and $\A\circ T_1$ respectively.
\end{lemma}

\begin{lemma}\label{lemma:invT}
	For each FST $T$ there exists an FST $T^{-1}$ that computes the inverse relation of the relation~$T$. FST $T^{-1}$ is log-space constructible by FST~$T$.
\end{lemma}

A \emph{rational cone} is a family of languages $\C$
that is closed under the rational dominance relation: 
$A \lerat B $ and $B \in \C$ imply $A \in \C$.
If there exists a language $F\in \C$ such that 
$L \lerat F$ 
for any $L \in \C$,
then $\C$ is a \emph{principal} rational cone generated by $F$; we denote it as $\C = \T(F)$.

Rational transductions for context-free languages were thoroughly investigated in the 1970s, particularly by the 
French school. The main results of this research were published in Berstel's book~\cite{BerstelBook}. As described in~\cite{BerstelBook}, 
it follows from the Chomsky-Sch\"utzenberger theorem
that $\CFL$ is a principal rational cone:   $\CFL = \T(D_2)$, where $D_2$ is the Dyck language on two types of brackets.

\subsection{Computational Models}

Firstly, we define BA. We provide the definition that is equivalent to the original definition from~\cite{BalloonHU67} but has technical differences, for the sake of convenience. Then we provide the definitions of other models: the refined definition of Balloon automata in terms of protocols and computational models based on $\logTM$ that are connected with $\nreg$-problem as well as with $1\NBPA$.

As it said, the balloon is a storage medium of unspecified structure. Thus its
states are represented by (a subset of) positive integers. A BA can
get limited information about the state of the balloon (the balloon
information function in the definition below) and can
modify the states of the balloon (the balloon control function).
Here we need 1BAs only. So we give the definition for them. The definitions for 2BAs are similar, they are provided in~\cite{BalloonHU67}.

\begin{definition}\label{def::BA}\definitionbodyfont 
A $1$-way balloon automaton \textup{(1BA)}
        is defined by a tuple $$\langle S, \Sigma_{\lendmarker\rendmarker}, B_S, B_I, \getbi, \updbs, F, s_0, \delta  \rangle,\quad\text{where}$$	
	\begin{itemize}
		\item $S$ is the finite set of automaton states.
		\item $\Sigma_{\lendmarker\rendmarker} = \Sigma\cup\{\lendmarker,\rendmarker\}$, where $\Sigma$ is the finite input alphabet and ${\lendmarker,\rendmarker}$ are the endmarkers. The input has the form $\lendmarker\! w\!\rendmarker,\, w\in \Sigma^*$.
		\item $B_S \subseteq \ZZ_{>0}$ is the set of the balloon states.
		\item $B_I$ is the finite set of the balloon information states.
		\item $\getbi: B_S \to B_I$ is a total computable function (balloon information function).
		\item $\updbs$ is a partially computable function from $ S\times B_S$ to $B_S$ (balloon control function).
		\item $F \subsetneq S$ is the set of the final states.
		\item $s_0 \in S\setminus F$ is the initial state.
		\item $\delta$ is the transition relation (a partial function for deterministic automata) defined as
                  $\delta \subseteq 	(S\times\Sigma_{\lendmarker\rendmarker, \eps}\times	B_I) \times S$; hereinafter $\Gamma_\eps = \Gamma \cup \{\eps\}$ for any alphabet $\Gamma$.
	\end{itemize}
\end{definition}

\begin{definition}\definitionbodyfont 
	A \emph{configuration} of a $1\BA$ is a triple $(q, u, i) \in S\times\Sigma_{\lendmarker\rendmarker}^*\times B_S$, where $u$ is the unprocessed part of the input $w$ so $u$ is either $\lendmarker\! w\! \rendmarker$ or a suffix of $w\! \rendmarker$. The \emph{initial configuration} of $1\BA$ is $(s_0, \lendmarker\! w\! \rendmarker, 1)$, a \emph{move} of $1\BA$ is defined by the \emph{relation} $\vdash$
on configurations  as follows:
	$ (q, \sigma u, i) \vdash (p, u, j)$,  where $\sigma \in \Sigma_{\lendmarker\rendmarker,\eps}$  if 
$j = \updbs(p,i),\, p \in \delta(q, \sigma, \getbi(i))$.
A $1\BA$ accepts the input $w$ if there exists a sequence of moves \emph{(computational path)} such that after processing of $\lendmarker\! w\! \rendmarker$ the final state is reached, i.e.,
$ (s_0, \lendmarker\! w\! \rendmarker, 1) \ders  (q_f, \eps, i)$,  where $q_f \in F,\, i \in B_S$. 
\end{definition}

It is not easy to define classes of balloon automata (like PDAs or SAs) 
since one needs to define valid families of functions $\getbi$ and $\updbs$. One can see an example of PDAs definition in terms of BA in~\cite{BalloonHU67}. We suggest another approach for the definition of BA classes in Section~\ref{sect:BA}. The  approach simplifies the definitions since it is only needed to define a language of correct protocols to define an ADS.

We define a protocol as a sequence of triples $p_i = u_i\q_i\r_i$ of the query-word $u_i$, the query $\q_i$ and the response~$\r_i$ on the query. Numerous extra conditions are listed in the following formal definition. 

\begin{definition}\label{def:PROT}\definitionbodyfont 	
	Let $\Gamma_{\wr},\Gamma_{\query}, \Gamma_{\resp}$ be finite disjoint alphabets such that $\Gamma_\query \neq \es, \Gamma_{\resp} \neq \es$.
	 Let $\valid \subseteq \Gamma_{\query}\times\Gamma_{\resp}$ be a relation that provides the correspondence between queries and possible responses. A \emph{protocol} is a word $p$ such that $p = p_1\cdots p_n$, where $n \geq 0$, $p_i = u_i\q_i\r_i$, $u_i \in \Gamma_{\wr}^*$, $q_i \in \Gamma_{\query}$,
$r_i \in \Gamma_{\resp}$, and $\r_i \in \valid(\q_i)$. We call a word $p_i$ a \emph{query block}. We say that a language $\PROT\subseteq (\Gamma_{\wr}^*\Gamma_{\query}\Gamma_{\resp})^*$ is \emph{a language of correct protocols} 
if the axioms \rm{(i-v)} hold:
\begin{enumerate}[leftmargin=+.5in,label=\rm{(\roman*)}]
\item\label{prot:axioms:empty} $\eps \in \PROT$;
	 	\item\label{prot:axioms:corr} $\forall p \in \PROT : p$ is a protocol;
		\item\label{prot:axioms:prefixes} $\forall p \in \PROT : $ if $p = p_1p_2$ and $p_1$ is a protocol, then $p_1 \in \PROT$;
		\item\label{prot:axioms:ext} $\forall p \in \PROT\; \forall u \in \Gamma_{\wr}^*\; \forall \q \in \Gamma_{\query}\; \exists \r \in \Gamma_{\resp} : pu\q\r \in \PROT$;
		\item\label{prot:axioms:uniqr} $\forall pu\q\r\in\PROT:$ if $p' \in\PROT$ and $p' = pu\q\r's$, then $\r' = \r$;
		\item\label{prot:axioms:reset} $\exists \q\in \Gamma_{\query}, \r \in \Gamma_{\resp} \; \forall p_1, p_2 \in \PROT : p_1\q\r p_2 \in \PROT $.
	 \end{enumerate}
\end{definition}

Axiom~\ref{prot:axioms:reset} does not hold in the general case, e.g., for SAs and counter automata without zero tests. It is needed to describe the connection of automata with an ADS with BAs in Section~\ref{sect:BA}.

A language of correct protocols $\PROT$ generates the corresponding class of languages, the principal rational cone~$\T(\PROT)$. All examples of BAs languages classes in~\cite{BalloonHU67} can be
presented as~$\T(\PROT)$. We provide here only two examples.

\begin{example}\label{ex:DPROT}\examplebodyfont
	It is well-known~\cite{BerstelBook} that $\CFL = \T(D_2)$, where $D_2$ is the Dyck language with two types
	of parentheses. It is also well-known that a Dyck word  is a protocol of the stack. We transform the language $D_2$ into a language of protocols $\DPROT$ as follows.
	
	We define the alphabets $\Gamma_{\wr} = \es$, $\Gamma_{\query} = \{ \push_\textsf(, \push_\textsf[, \pop \}$, $\Gamma_\resp = \{ \textsf(, \textsf), \textsf[, \textsf] \}$, $\valid = \{ (\push_\textsf[, \textsf[\,), (\push_\textsf(, \textsf(\,), (\pop, \textsf]\,), (\pop, \textsf)\,) \}$.
	To define correct protocols we use an FST $T$ that erases all symbols from $\Gamma_{\query}$ of the input. 
	So, $$\DPROT = \{ p \mid T(p) \in D_2\}.$$

	By the definition $D_2 \lerat \DPROT$, so we have that $\T(D_2) \subseteq \T(\DPROT)$. 
	It is also easy to show that $\DPROT \lerat D_2$, so $\T(\DPROT)= \T(D_2) = \CFL$.
	
	Note that we set here $\Gamma_{\wr} = \es$ for the sake of simplicity. One can use another variant: 
	$\Gamma_{\wr} = \{ \textsf(, \textsf[ \}$, $\Gamma_{\query} = \{ \textsf{multipush}, \pop \}$,
	$\Gamma_{\resp} = \{ \textsf{pushed}, \textsf), \textsf] \}$.
\qed\end{example}

The following example is a starting point  for  the generalization presented in this paper.
\begin{example}\label{ex:SAPROT}\examplebodyfont	
	The data structure Set consists of the set $\SS$ which is initially empty. Set supports the following operations: $\ins(x): \SS \to \SS\cup\{x\}$, $\out(x): \SS \to \SS\setminus\{x\}$, $\test(x): x \stackrel{?}{\in} \SS$. We define the protocol language $\SAPROT$ consistently with~\cite{RV17,RVCSR2018}, so the elements of alphabets below are individual symbols while they are words in~\cite{RV17,RVCSR2018}. $\Gamma_\wr = \{a,b\}$, $\Gamma_{\query} = \{\#\ins, \#\out, \#\test \}, \Gamma_{\resp} = \{\#, +\#, -\#\}, \valid = \{ (\#\ins,\#), (\#\out,\#), (\#\test,+\#), (\#\test,-\#)\}$. 
	
	It was proved in~\cite{RV17} that $\CL(1\mathrm{NSA}) = \T(\SAPROT)$. 	
        \qed
\end{example}

\begin{definition}\label{def:BAbyPROT}\definitionbodyfont  
	Fix a language of correct protocols $\PROT$. An \emph{automaton equipped with auxiliary data structure} $\BP$ (defined by $\PROT$) is defined by a tuple
	$$ \langle S, \Sigma_{\lendmarker\rendmarker}, \Gamma_{\wr}, \Gamma_{\query}, \Gamma_{\resp}, F, s_0, \delta  \rangle, \text{ where} $$
		\begin{itemize}
			\item $S$, $\Sigma_{\lendmarker\rendmarker}$, $F$, $s_0$ are the same as in Definition~\ref{def::BA}, so as $\Sigma_{\lendmarker\rendmarker, \eps}$.
			\item $S = S_{\wr} \cup S_{\query}$, $ S_{\wr} \cap S_{\query} = \es$.
			\item $\PROT \subseteq (\Gamma_{\wr}^*\Gamma_{\query}\Gamma_{\resp})^*$.
			\item $\delta$ is the transition relation defined as 
			$$\delta \subseteq  ([S_\wr \times \Sigma_{\lendmarker\rendmarker, \eps}] \times [\Gamma_{\wr}^* \times S])
					 \cup		(S_\query \times \Gamma_{\query} \times \Gamma_{\resp} \times S_\wr).$$		
		\end{itemize}
	The automaton has a one-way write-only query tape. During the processing of the input, it writes query-words $u_i \in \Gamma_{\wr}^*$ on the query tape, performs queries $\q_i$, and receives responses $\r_i$ such that
	 $u_1\q_1\r_1 \cdots u_n\q_n\r_n \in \PROT $.
	 After each query, the query tape is erased.
	 
	 A \emph{configuration} of an ADS-automaton is a tuple 
	 $$(s, v, u, p) \in S\times\Sigma_{\lendmarker\rendmarker}^*\times\Gamma_{\wr}^*\times(\Gamma_{\wr}^*\Gamma_{\query}\Gamma_{\resp})^*,$$ where $v$ is the unprocessed part of the input~$w$, i.e., $v$ is the suffix of $\lendmarker\!w\!\rendmarker$, $u$ is the content of the work tape, and $p$ is the protocol of the automaton operating with the data structure. 
	 A move of an automaton is defined via the relation $\vdash$ on configurations which is defined as follows:
	 \begin{align}
	   (s, av, u, p) &\der (s', v, ux, p), & &\text{ if } s \in S_\wr,\, (s, a, x, s') \in \delta, \label{eq:wrmove}\\
	   (s, v, u, p) &\der (s', v, \eps, pu\q\r), & &\text{ if } s \in S_\query,\, (s, \q, \r, s') \in \delta,\, pu\q\r\in \PROT. \label{eq:qmove}
	 \end{align}
	A configuration is \emph{initial} if it has the form $(s_0, \lendmarker\!w\!\rendmarker, \eps, \eps)$, a configuration is \emph{accepting} if it has the form $(s_f, \eps, \eps, p)$, where $s_f \in F, p \in \PROT$. A word $w$ is \emph{accepted} by an automaton with ADS if $(s_0, \lendmarker\!w\!\rendmarker, \eps, \eps) \ders (s_f, \eps, \eps, p)$. An automaton is \emph{deterministic} if for all configurations $c, c_1, c_2$ from $c \der c_1$ and $c \der c_2$ follows $c_1 = c_2$.
\end{definition}

For the next two models, we provide the definitions on the implementation level only. 

\begin{definition}\label{def:NBPlogTM}\definitionbodyfont 
 A $\DBP\logTM$ ($\NBP\logTM$) is a deterministic {(nondeterministic)} $\logTM$~$M$ equipped with an ADS defined by the language of correct protocols~$\PROT$. I.e., $M$ is equipped with an additional write-only one-way query tape that is used to write down a query word $u_i$ and perform a query. After a query $\q_i$ is performed, the tape is erased and the finite state control of $M$ receives the result $\r_i$ of the query~$\q_i$. The query results are consistent with $\PROT$, i.e., $p_1\cdots p_n \in \PROT$, $p_i = u_i\q_i\r_i$. 
 
A \emph{configuration} of $\BP\logTM$ is a triple~$(c, u, p)$ where $c$ is the configuration of $\logTM$-part, $u$ is the word written on the query tape, and $p\in\PROT$ is the protocol that is the result of all the performed queries. A $\BP\logTM$~$M$ accepts a word $w$ if $(c_0(w), \eps, \eps) \ders (c_f, \eps, p)$, where $c_0(w)$ is the initial configuration of the $\logTM$-part of $M$, $c_f$ is the accepting configuration of $\logTM$-part of $M$, and $p\in \PROT$, the relation $\der$ corresponds to the $M$'s moves.
\end{definition}

\begin{definition}\label{def:NAFlogTM}\definitionbodyfont 
 Let $F$ be an arbitrary formal language (filter). A $\DA_{F}\logTM$ ($\NA_{F}\logTM$) is a deterministic (non-deterministic) log-space TM equipped with a read-only one-way infinite tape called \emph{advice tape}. At the beginning of the computation, the advice tape contains a word $y\Lambda^{\infty}$, where $y\in F$ and $\Lambda$ is a symbol that indicates empty cells.
 
A \emph{configuration} of an $\mathrm{A}_{F}\logTM$ $M$ is a pair~$(c, u)$ where $c$ is the configuration of the $\logTM$-part of $M$, $u$ is the unprocessed part of $y$. $M$ accepts a word $x$ if there exists $y \in F$ such that $(c_0(x), y) \ders (c_f, \eps)$, where $c_0(x)$ is the initial configuration of the $\logTM$-part of $M$, $c_f$ is the accepting configuration of the $\logTM$-part of $M$.
\end{definition}

An equivalent model to $\DA_{F}\logTM$s appeared in~\cite{Vya09CSR} and its journal version~\cite{VyaPIT11} under the name ``models of generalized nondeterminism (GNA)''
and lead to the appearance of the $\reg(F)$ problem. In this paper we repeat the steps of~\cite{Vya09CSR,VyaPIT11} to establish the connection between $\NA_{F}\logTM$ and $\nreg(F)$ problem in
Section~\ref{sect:Models}
to prove one of the main results of the paper Eq.~\eqref{eq:NRRlogTMresult}. We also show the equivalence between $\DA_{F}\logTM$s and GNA. The difference is that in GNA it is allowed not to process the advice till the end of the word, so it was demanded for $F$ to be a prefix-closed language in~\cite{Vya09CSR,VyaPIT11}.

\section{Principal Rational Cones and\\ the NRR-Problem}\label{sect:Cones}

In this section, we provide the core of our technique. We prove that $\CL(1\NBPA)$ is a principal rational
cone generated by the language of correct protocols~$\PROT$, i.e., $\CL(1\NBPA) = \T(\PROT)$; it is the first main result of the section. 
This fact yields structural results about the family~$\CL(1\NBPA)$, as well as the results on the complexity of the emptiness problem.  
We focus in this section on the connection between the non-emptiness problem and the $\nreg(\PROT)$ problem. We prove that these problems are equivalent under log-space reductions, it is the second main result of the section. 
It leads us to the main results of the paper in
Section~\ref{sect:Models}.
We provide in this section structural results that naturally arise in the proofs. Other structural results are discussed in Section~\ref{sect:BA} since their relation to~\cite{BalloonHU67}.

Most of the results of this section directly generalize the results of~\cite[Section~3]{RV17} (see the full journal version~\cite{RV21}). 
In most cases, to get a generalized result, one can substitute SA protocols (see Example~\ref{ex:SAPROT}) with general protocols as defined in Definition~\ref{def:PROT}. 
So, our general approach comes from the generalization of the technique that was developed for SAs. One can also find in~\cite{RV21} more technically detailed proofs.

\begin{lemma}\label{lemma:MPROT}
	There exists a $1\NBPA$ $M_\PROT$ recognizing $\PROT$.
\end{lemma}
\begin{proof}
  Let us assume that $M_\PROT$ has on the input the word of the form $p_1\cdots p_n$, where $p_i = u_i\q_i\r_i$ (since it is a regular condition). $M_\PROT$ writes a word $u_i$ on the query tape, performs the query $\q_i$ and tests that the responce is $\r_i$. If all tests are correct than $p$ is
  accepted; otherwise, it is rejected.
\end{proof}

\begin{lemma}\label{lemma:TwoLettersProt}
	For each language of correct protocols $\PROT \subseteq (\Gamma_{\wr}^*\Gamma_{\query}\Gamma_{\resp})^* $ there exists a language of correct protocols $\PROT_{\{a,b\}} \subseteq (\{a,b\}^*\Gamma_{\query}\Gamma_{\resp})^* $, provided $(\Gamma_{\query}\cup\Gamma_{\resp}) \cap \{a,b\} = \es$, such that the following properties hold
	\begin{itemize}
		\item $\CL(1\NBPA) = \CL(1\mathrm{NB}_{\PROT_{\{a,b\}}}\mathrm{A})$,
		\item $\CL(1\DBPA) = \CL(1\mathrm{DB}_{\PROT_{\{a,b\}}}\mathrm{A})$,
		\item There exists a DFST $T$ such that $T(\PROT) = \PROT_{\{a,b\}}$ and $T^{-1}$ is a DFST,
		\item For each $1x\BPA$ $M$ there exists an equivalent $1x\mathrm{B}_{\PROT_{\{a,b\}}}\mathrm{A}$ $M_{\{a,b\}}$ such that $M_{\{a,b\}}$ is log-space constructible by $M$ and vice versa ($x \in \{N, D\}$).
	\end{itemize}	
\end{lemma}
\begin{proof}
	Enumerate all letters from $\Gamma_{\wr}$ and encode the $i$-th letter as $ab^ia$. Such encoding is computable by a DFST $T$ and the inverse encoding is computed by $T^{-1}$ that is a DFST as well. So, $\PROT_{\{a,b\}} = T(\PROT)$ (we assume that $T$ preserves letters from $\Gamma_{\query}\cup\Gamma_{\resp}$). Now we show that for each $1\NBPA$ $M$ there exists an equivalent $1\mathrm{NB}_{\PROT_{\{a,b\}}}\mathrm{A}$ $M_{\{a,b\}}$.
	
	By our construction, $M_{\{a,b\}}$ simulates $M$, i.\,e., $M_{\{a,b\}}$ has states of the form $(s, \aux)$ where $s$ is a state of $M$ and $\aux$ is an auxiliary information needed for simulation; so for each configuration $((s, \aux), v, u', p')$ of $M_{\{a,b\}}$ there is a corresponding configuration $(s, v, u, p)$ of $M$, where $p' = T(p)$ and $u'$ is a prefix of $T(u)$. $M_{\{a,b\}}$ computes $T(u)$ by simulation of $T$ via finite control and information of this simulation is stored in $\aux$; the other part of finite control simulates $M$'s transitions.
	
	Each $1\mathrm{NB}_{\PROT_{\{a,b\}}}\mathrm{A}$ $M_{\{a,b\}}$ can be simulated by a $1\NBPA$ $M$ in the same way, one shall use $T^{-1}$ instead of $T$. Note that described simulations preserve determinism and the transformations between $M$ and $M_{\{a,b\}}$ are log-space computable.
\end{proof}

\begin{lemma}\label{lemma:TcompM}
Let $T$ be an FST with the input alphabet $\Delta$ and the output alphabet $\Sigma$ and $M$ be a $1\NBPA$ over the alphabet~$\Sigma$. There exists a $1\NBPA$ $ M' = M\circ T$ recognizing the language  $T^{-1}(L(M))$. If $T$ is a DFST and $M$ is a $1\DBPA$ then $M'$ is a $1\DBPA$ as well.   
\end{lemma}
\begin{proof}
	The simulation is performed in a straight-forward way. $M'$ guesses an image $w\in T(w')$ of the input word $w'$ such that $w \in L(M)$ if $T(w')\cap L(M) \neq \es$, computes $w$ by simulation of $T$ and simulates $M$ on the input $w$. $M'$ has configurations of the form $((s, \aux), v', u, p)$ that correspond to configurations $(s, v, u, p)$ of $M$. As in the proof of Lemma~\ref{lemma:TwoLettersProt}, the $\aux$ information is used to simulate $T$ via finite state control. The construction preserves determinism of $1\DBPA$ if $T$ is a DFST.
\end{proof}

\begin{lemma}\label{lemma:ProtExtr}
	Let $M$ be a $1\NBPA$. There exists an FST $T_M$ such that $w\in L(M) $ iff $ T_M(w) \cap \PROT \neq \es$.
	Moreover, $p \in T_M(w)$ iff $M$ has a run on $w$ such that $(s_0, w, \eps, \eps) \ders (s_f, \eps, \eps, p)$.
\end{lemma}

We denote by $s \xrightarrow[x]{a} s'$ the move of $T_M$ from the state $s$ to the state $s'$ on which it reads $a$ from the input tape and writes $x$ on the output tape.

\begin{proof}
	One can construct $T_M$ by $M$ as follows.
	$T_M$ has the same states as $M$ (and the same initial state and set of accepting states).
	In the case of move~\eqref{eq:wrmove}, $T_M$ has the move
        $s \xrightarrow[x]{a} s'$, 
	and in the case of move~\eqref{eq:qmove}, $T_M$ has moves $s \xrightarrow[\q\r']{\eps} s'$ for all $\r'$ such that $(s, \q, \r', s') \in \delta_M$. 
	
	Assertion $p \in T_M(w)\cap \PROT$ implies that $M$ has the corresponding run by axiom~\ref{prot:axioms:uniqr} in Definition~\ref{def:PROT}. So, if $T_M(w)$ contains a correct protocol $p$ then $M$ has the run $(s_0, w, \eps, \eps) \ders (s_f, \eps, \eps, p)$. The implication in the other direction directly follows from the construction of $T_M$.
\end{proof}

\begin{definition}\definitionbodyfont 
	An FST $T_M$ from Lemma~\ref{lemma:ProtExtr} called \emph{extractor} (of protocols).
\end{definition}

\begin{theorem}\label{th:BPisPCone}
	 $\CL(1\NBPA) = \T(\PROT)$.
\end{theorem} 
 \begin{proof}
 	Lemma~\ref{lemma:ProtExtr} implies that for each $1\NBPA$ $M$ there exists an extractor $T_M$ such that  $L(M) = T_M^{-1}(\PROT)$, and by Lemma~\ref{lemma:invT} there exists FST~$T = T_M^{-1}$ such that $L(M) = T(\PROT)$, so $L(M) \lerat \PROT$ and therefore $\CL(1\NBPA) \subseteq \T(\PROT)$.
	
 	The inclusion $\T(\PROT) \subseteq \CL(1\NBPA)$ follows from Lemmata~\ref{lemma:MPROT} and~\ref{lemma:TcompM}: for each $L = T'(\PROT)$ we take an FST $T = T'^{-1}$ and apply the lemmata.
 \end{proof}

\begin{theorem}\label{th:RReqNonEmp}
	$\ov{\EP1\NBPA} \lelog \nreg(\PROT) \lelog \ov{\EP1\NBPA}$.
\end{theorem}
\begin{proof}
	Let $M$ be the input of the non-emptiness problem $\ov{\EP1\NBPA}$ 
	and $T_M$ be the corresponding extractor. By Lemma~\ref{lemma:ProtExtr},
	$w \in L(M) \iff T_M(w) \cap \PROT \neq \es$. 
	So, $L(M) \neq \es \iff T_M(\Sigma^*) \cap\PROT\neq\es$.
	Construct an NFA $\A$ recognizing $T_M(\Sigma^*)$ by Lemma~\ref{lemma:EM} in log space.
	So, $$L(M) \neq \es \iff L(\A) \cap\PROT\neq\es \stackrel{\text{Def.~\ref{def:RR}}}{\iff} \A \in \nreg(\PROT).$$
	So we have proved $\ov{\EP1\NBPA} \lelog \nreg(\PROT)$.
	
	The reduction $\nreg(\PROT) \lelog \ov{\EP1\NBPA}$ follows from Lemmata~\ref{lemma:MPROT} and~\ref{lemma:TcompM}. We construct by $\A$ on the input of $\nreg(\PROT)$ the automaton $M = M_\PROT \circ T$, where $xTy \iff (x = y) \land (x \in L(\A))$.
\end{proof}

\begin{theorem}\label{th:MP1DPBAtoDRR}
	$\MP1\DBPA \lelog \reg(\PROT)$.
\end{theorem}
\begin{proof}
	We construct a DFA $\A$ on the input of $\reg(\PROT)$ by $(w, M)$ on the input of $\MP1\DBPA$ via a log-space transducer. The idea is that $\A$ simulates $M$'s run on the input $w$ and checks the correctness of the protocol by reading the input word, that is a protocol $p \in \PROT$. The protocol~$p$ is accepted iff $p$ is the protocol of $M$ on processing of $w$ and $M$ accepts $w$.

        A state of $\A$ is a tuple $(s, i, \aux)$ where $s\in S_M$, $i$ is the index of the letter $w_i$ over the $M$'s head and $\aux$ is the auxiliary information needed for the simulation.
To simulate a transition of $M$, the following actions are performed by $\A$. If $M$ writes a word $v$ (a subword of the future query word) to the query tape, $\A$ stores $v$ in the finite memory (a part of $\aux$ component of its states) and checks whether the unprocessed part of its input begins with  $v$  (if not, the input word $p$ is rejected). If  $M$ performs a~query $\q$, $\A$ verifies that the unprocessed part of its input begins with  $\q\r$ and performs the transition that $M$ does after receiving $\r$ as a response. $\A$ accepts the input $p$ if it was not rejected during the simulation, $i = |w|+1$ (i.e., $M$'s head is over $\rendmarker$) and $M$ is in accepting state.

It follows from the construction 
that $\A$ accepts $p$ iff $p$ is the protocol of $M$ processing the input $w$. Note that this protocol is unique since $M$ is a $1\DBPA$. Also since $M$ is $1\DBPA$, $\A$ is log-space constructible. Finally, $L(\A) \cap \PROT \neq \es \iff w \in L(M)$, so $\MP1\DBPA \lelog \reg(\PROT)$.
\end{proof}

\section{Connection with Balloon Automata}\label{sect:BA}

We provide a high-level description of 
classes~$\CM_B$ of BAs.
The definition in a more formal style could be found in~\cite{BalloonHU67}.

\begin{definition}\label{def:BAclass}\definitionbodyfont 
	A subset of BAs $\CM_B$ is a \emph{class of BAs} if the following conditions hold. 	
	\begin{enumerate}[leftmargin=+.5in,label={(\Roman*)}]
	\item\label{def:BAclass:reset}  $\CM_B$ contains all automata
          with $\getbi$  such that, for each state~$s$,\linebreak
          $\updbs(s,i)$ is either $i$ for all $i$ or $\updbs(s,i) = j$ for all $i$ and some constant~$j$.
		\item\label{def:BAclass:combine}  If $\A, \B \in \CM_B$, $\updbs^\A, \getbi^\B $, $\updbs^\A, \updbs^\B$ are the corresponding functions of $\A$ and $\B$, then $\CM_B$ includes each automaton $\CC$ such that $\getbi^\CC$ and $\updbs^\CC$ are the functions that are obtained from the functions of $\A$ and $\B$ via finite control, i.e., for each state $s \in S_\CC$ $\getbi^\CC(s, i)$ equals to either $\getbi^\A(s,i)$ or $\getbi^\B(s,i)$ for all $i$, for each $i, j$ if $\updbs^\CC(i) \neq \updbs^\CC(j)$ then either $\updbs^\A(i) \neq \updbs^\A(j)$ or $\updbs^\B(i) \neq \updbs^\A(j)$.
	\end{enumerate}
\end{definition}	

Property~\ref{def:BAclass:reset} implies that $\CM_B$ contains automata that can reset
any state~$i$ of the balloon to the initial state~$1$ (or to some fixed state $j$ as well). Together with Property~\ref{def:BAclass:combine} it implies that the balloon has a reset operation that sets the balloon's state to the initial state. This property does not hold for SAs, so there is no direct correspondence between
classes of languages of BAs and automata with an ADS in the general case. 

\begin{theorem}
	For each ADS $\BP$ there exists a balloon $\mathrm B$ and a subset of BAs $\CM_B$ such that the corresponding  classes of languages coincide, i.e. $\CL(xy\BPA) = \CL(xy\mathrm{BA})$ and     Property~\ref{def:BAclass:combine} holds. If $\PROT$ has the reset operation~\ref{prot:axioms:reset}, Property~\ref{def:BAclass:reset} also holds, i.e $\CM_B$ is a class (in  terms of Definition~\ref{def:BAclass}).
\end{theorem}
\begin{proof}
  We begin with the construction of the balloon~$B$ and the bijection from $xy\BPA$ to $xy\mathrm{BA}$ such that $xy\mathrm{BA}$ form the set $\CM_B$
  satisfying Property~\ref{def:BAclass:combine}.
  A state of the balloon $B$ is an integer that is the encoding of pairs of words $(p, u)$, where $p$ is the current protocol, i.e., the protocol of all previous operations before the upcoming move, and $u$ is the word on the query-tape. We enumerate all $p \in \PROT$ and all $u\in \Sigma^*$ and use the standard enumeration of pairs of integers.
	
	Firstly, we define functions $\updbs$ and $\getbi$ for the BA $M^B_\PROT$ recognizing $\PROT$. Recall that for any language of correct protocols~$\PROT$ there exists $1\DBPA$ $M_\PROT$ recognizing $\PROT$ by Lemma~\ref{lemma:MPROT}. 
The function $\updbs$ simulates write operations and queries: it just updates the ballon's state according to the encoding. The function $\getbi : \ZZ_{>0} \to \Gamma_{\resp}\cup\{\bot\}$ returns responses or $\bot$ if there were no query. For an arbitrary $xy\BPA$ $M$ the $xy\BA$ $M^B$ is constructed as follows. The function $\getbi^M$ is the same as for $M^B_\PROT$ for any $M^B$. The function $\updbs^M$ is a modification of the function for $M^B_\PROT$ according to the finite state control of $M$. We define the class $\CF(\updbs)$ more formally below to show that Property~\ref{def:BAclass:combine} holds.

As the result, the states of the ballon~$B$ just encode the part of $xy\BPA$ that describes the data structures, and $\updbs$ and $\getbi$ simulate the work with the data structure defined by $\PROT$. So we provided the bijection between $xy\BPA$ and $xy\mathrm{BA}$.

We move to a formal definition of the class $\CF(\updbs)$.
At first, we define  $\CF_1(\updbs)$ satisfying Property~\ref{def:BAclass:combine}.
Assume that $M_\PROT$ writes at most one letter to the query tape per move and it also has a state $s_\eps$  in which it neither writes nor performs query. So, $\updbs(s_\eps, i) = i$ for all $i$ (hereinafter in $M_\PROT^B$). We mark a state $s$ as $s_a$ if $M_\PROT$ writes $a$ on the output tape and mark a state $s$ as $s_\q$ if $M$ performs query $\q$. From definition follows that $\updbs(s_m, i) = \updbs(s'_m, i)$ for all states $s$ and $s'$ marked by the same mark ($a$, $\eps$ or $\q$). So we define a function $\updbs^\PROT : (\Gamma_{\wr,\eps}\cup \Gamma_\query)\times \ZZ_{>0} \to \ZZ_{>0}$ so that $\updbs^\PROT(m, i) = \updbs(s_m, i)$. So for any $xy\mathrm{BA}$ $M$ the function $\updbs^M$ defined as follows. The states of $M$ are marked by symbols from $\Gamma_{\wr,\eps}\cup \Gamma_\query$ and $\updbs^M(s_m, i) = \updbs^\PROT(m, i)$. It is easy to see that from our definition of
$\CF_1(\updbs)$
follows the bijection between $xy\BPA$ and $xy\mathrm{BA}$ and Property~\ref{def:BAclass:combine} holds as well.

If $\PROT$ has the reset operation~\ref{prot:axioms:reset}, then we shall modify the interpretation of $B$ since $\CF(\updbs)$ does not satisfy our definition anymore. Firstly we describe the interpretation of $\updbs$ functions for $xy\mathrm{BA}$ $M^B$ from Property~\ref{def:BAclass:reset}. If $\updbs(s,i) = j$, $i$ and $j$ encode pairs $(p_i, u_i)$ and $(p_j, u_j)$ respectively and $(p_j, u_j)$ is not obtained from $(p_i, u_i)$ by a single move of $M^B_\PROT$, then we interpretate the state change $i\to j$ as follows. The corresponding to $M^B$ $xy\BPA$ $M$ performs the reset operation and then performs sequence of queries that move the configuration from $(\eps, \eps)$ to $(p_j, u_j)$ during $\eps$ moves. Note that by the definition of Property~\ref{def:BAclass:reset}  $M^B$ has finitely many $j$'s in the range of $\updbs$ so $M$ is well-defined. Denote $\updbs$ functions for automata from Property~\ref{def:BAclass:reset} as $\CF_I(\updbs)$. So $\CF(\updbs)$ is a closure of $\CF_I(\updbs)$ and $\CF_1(\updbs)$ in terms of Property~\ref{def:BAclass:combine}. From Property~\ref{def:BAclass:combine} and our construction of $xy\BPA$'s for Property~\ref{def:BAclass:reset} follows the construction of $xy\BPA$ for any of $xy\BA$ from the closure in terms of Property~\ref{def:BAclass:combine}. So, we have proved the second part of the theorem.
\end{proof}

So all the results from~\cite{BalloonHU67} that do not rely on~\ref{def:BAclass:reset} hold for $B_\PROT$-automata. We are most interested in~\eqref{eq:HUresults} and its complexity analogue~\eqref{eq:HUresultsLog}. Many structural results from~\cite{BalloonHU67}
follow from the fact that $\CL(1\NBPA)$ is a principal cone (Theorem~\ref{th:BPisPCone}), namely, closure
of $\CL(1\NBPA)$ over union and rational transductions\footnote{Intersection and quotient with regular languages, gsm forward mapping are the partial cases of rational transductions.}. We shall also mention the closure over gsm inverse mappings proved in~\cite{BalloonHU67} for all $xy\BA$ that implies the same closure for all $xy\BPA$.

\begin{lemma}\label{lemma:ResetOpClosure}
	If $\BP$ contains the reset operation then $\CL(1\NBPA)$ is closed over concatenation and iteration.
\end{lemma}
\begin{proof}
	We construct $1\NBPA$'s $M^2$ and $M^*$ by $1\NBPA$ $M$ recognizing $L(M)\cdot L(M)$ and $L(M)^*$ respectively in a straight-forward way. $M^2$ simulates $M$ and nondeterministically guess the split of the input $uv$ such that $u, v \in L(M)$ at the end of the $u$. If after processing of $u$, $M$ is in an accepting state, $M^2$ performs the reset operation and simulates $M$ on $v$. $M^*$ guesses the split of the input into $u_1u_2\ldots u_m, u_i \in L(M)$ and acts in the similar way.
\end{proof}

The standard technique from~\cite{BerstelBook} implies the following lemma. 

\begin{lemma}
	If $\PROT\#\PROT \lerat \PROT$, $\# \not \in \Gamma$, then $\CL(1\NBPA)$ is closed over concatenation.
	If $(\PROT\#)^* \lerat \PROT$, $\# \not \in \Gamma$, then $\CL(1\NBPA)$ is closed over iteration.
\end{lemma}
\begin{proof}[Proof idea] 
	The construction is similar to the one from the proof of Lemma~\ref{lemma:ResetOpClosure}. FSTs $T_{L^2}$ and $T_{L^*}$ uses marks $\#$ to split the input and simulate an FST $T_L$ corresponding to the language $L$, i.e., $L = T_L(\PROT)$.
\end{proof}

\begin{remark}\remarkbodyfont
	We leave open the question of the reduction in the opposite direction. I.e., does for each class of BAs exist a language of correct protocols~$\PROT$ such that BAs recognize the same class of languages as $\BP$ automata?  The essence of the problem is as follows. If axioms (I-II) for the class of BAs are satisfied, does it imply that there exists a ``universal'' BA~$M_U$ such that  $\CL(1\NBA) = \T(L(M_U))$ and for each $M \in 1\NBA$ there exists an FST $T$ such that $L(M) = L(M_U\circ T)$? 
\end{remark}

\section{RR Problems and \boldmath$\logTM$ Models}\label{sect:Models}

\subsection{\boldmath$\mathrm{A}_F\logTM$ models}

In~\cite{Vya09CSR,VyaPIT11} it was shown that $\reg(\Pref(F\Lambda^*))$ is a complete problem in the class 
of languages recognizable by GNA with advices from~$F$ (this model corresponds to $\DA_{F}\logTM$, $\Pref(L)$ is the set of all prefixes of $L$). We prove that $\reg(\Pref(F\Lambda^*))$ and $\reg(F)$ are complete problems in the class $\CL(\DA_{F}\logTM)$. We begin with the proof of similar result for $\nreg(F)$ and $\CL(\NA_{F}\logTM)$. 
We introduce the following auxiliary lemma for the sake of the proof.

\begin{lemma}[\!\!\cite{RVRR2015DCFS}]\label{lemma:leratRR}
$F_1 \lerat F_2 \Rightarrow \nreg(F_1) \lelog \nreg(F_2)$.
\end{lemma}

\begin{lemma}\label{lemma:NAFlogTMtoNreg}
	 $\CL(\NA_F\logTM) \lelog \nreg(F)$.
\end{lemma}
\begin{proof}
	An $\NA_F\logTM$ $M$ takes on the input a word $x$ and also takes $y\in F$ on the advice tape. $x \in L(M) \iff \exists y \in F : M(x, y) = 1$, where $M(x, y) = 1$ if $M$ halts in an accepting state, $M(x, y) = 0$ if $M$ halts in a rejecting state.
	
	A surface-configuration of a $\NA_F\logTM$ is a tuple $(q, \mem, i, j)$ of the state $q$, log-space memory configuration $\mem$ and the positions $i$ of the head on the input tape and $j$ of the head on the advice (one-way) tape. The tuples $(q, \mem, i)$ are the states of finite automata $\A$ on the input of $\nreg(F)$ problem constructed by $M$ and $x$. The initial state is
        $(q_0, \es, 0)$, where $q_0$ is the initial state of the $\NA_F\logTM$,
        accepting states are states of the form $(q_f, \mem, i)$ where $q_f$ is an accepting state of the $\NA_F\logTM$. The transitions between the states are determined by letters of $y$, i.e. $(q, \mem, i, j) \der (q', \mem', i', j')$, $i' \in \{i-1, i, i+1\}$, $j' \in \{j, j+1\}$ if $(q', \mem', i') \in \delta_\A((q, \mem, i, j), y_j)$. The list of the $\A$'s transitions $\delta_\A$ is log-space computable so as the set of the $\A$'s states as well.

	Without loss of generality, we assume that $M$ always processes $y$ till the end, i.e. till mits $\Lambda$ on the advice tape. So, by the  construction of $\A$, we obtain 
\begin{align*}
 x \in L(M) & \iff  \exists y \in F: M(x,y) \iff \exists y \in F, k \geq 0: y\Lambda^k \in L(\A) \iff\\
 &\iff \A \in \nreg(F\Lambda^*)  \stackrel{\text{Lemma~\ref{lemma:leratRR}}}{\iff} \A' \in \nreg(F),
\end{align*}
where $\A'$ is constructed by $\A$ due to the reduction in Lemma~\ref{lemma:leratRR}. Since $M$ is fixed, we get that  
$$x \stackrel{?}{\in}L(M) \lelog \A \stackrel{?}{\in} \nreg(F\Lambda^*) \lelog \A' \stackrel{?}{\in} \nreg(F),$$
and we obtain that $\CL(\NA_F\logTM) \lelog \nreg(F)$ by the transitivity of the log-space reductions.  
\end{proof}	

\begin{remark}\remarkbodyfont
	Note that $F \simrat F\Lambda^* \simrat \Pref(F\Lambda^*)$ so the $\nreg$ problems for these filters are equivalent (up to $\lelog$ reductions). The equivalence holds since nondeterministic FST's can have several images for the same word, particularly write many $\Lambda$'s at the end of the word. It does not hold for deterministic FSTs, so $F\Lambda^* \xcancel{\leq}_{\mathrm{drat}} F$. So to obtain the corresponding lemma for $\DA_F\logTM$ we need to modify the proof of Lemma~\ref{lemma:NAFlogTMtoNreg}. 
\end{remark}

\begin{lemma}\label{lemma:DAFlogTMtoDreg}
	 $\CL(\DA_F\logTM) \lelog \reg(F)$.
\end{lemma}
\begin{proof}
	We repeat the steps of the proof of Lemma~\ref{lemma:NAFlogTMtoNreg}. Note that $\A$ is a DFA, since $M$ is a deterministic machine. To construct $\A'$ we construct an auxiliary DFA $\A''$ by $\A$ as follows.  $\A''$ has the states $(q, \bcancel{\Lambda})$ and $(q, \Lambda)$ for each state $q$ of $\A$. The auxiliary bit of a state indicates whether $\A''$ met $\Lambda$. So for each transition $q \xrightarrow{\Lambda} p$ of $\A$ there are two corresponding transitions $(q, b) \xrightarrow{\Lambda} (p, \Lambda)$, $b \in \{\Lambda, \bcancel{\Lambda}\}$. For states $(q, \Lambda)$ $\A'$ has only transitions by $\Lambda$. For the transitions $q \xrightarrow{a} p$, $a \neq \Lambda$, $\A''$ has transitions $(q, \bcancel{\Lambda}) \xrightarrow{a} (p, \bcancel{\Lambda})$. A state $(q, b)$ is an accepting if $q$ is an accepting state of $\A$.
	
Now we construct $\A'$. It is obtained from $\A$ by removing all $\Lambda$-transitions. Each $\A$'s accepting state is an accepting for $\A'$ and $\A'$ also has accepting states determined as follows. If $(q, \bcancel{\Lambda})$ has $\Lambda$-path to an accepting state $(q_f, \Lambda)$ in $\A''$, then $q$ is an accepting state in $\A'$. It is easy to see that $y \in L(\A') \iff \exists k \geq 0\; y\Lambda^k \in L(\A)$ and $\A''$ is log-space computable from $\A$ and $\A'$ is log-space computable from $\A$ and~$\A''$.
\end{proof}

\begin{lemma}\label{lemma:NregToNAFlogTM}
  $\nreg(F) \lelog \CL(\NA_F\logTM)$\hspace*{-0.15pt} and $\reg(F) \lelog \CL(\DA_F\logTM)$.
  Moreover, there exist an $\NA_F\logTM$ $M_\nreg$ that recognizes the problem~$\nreg(F)$ and $M_\reg$ that recognizes $\reg(F)$ as well. 
\end{lemma}
\begin{proof}
  The proof is straightforward. $M_\nreg$ gets on the input an NFA $\A$ and verifies whether $\A$ accepts  the word $y \in F$ written on the advice tape. If $y \in L(\A)$, $M_\nreg$ nondeterministically guesses the $\A$'s run on $y$. So, by Definition~\ref{def:NAFlogTM}, $\A \in L(M_\nreg)$ iff $\exists y \in F : y\in L(\A) \iff \A \in \nreg(F)$.
  
  The construction for $M_\reg$ is the same.
\end{proof}

\begin{theorem}\label{th:CLsRRsummary}
  \[
  \begin{aligned}
  &\CL(\NA_F\logTM) = \{ L \mid L \lelog \nreg(F) \},\\
  &\CL(\DA_F\logTM) = \{ L \mid L \lelog \reg(F) \}.    
  \end{aligned}
  \]
\end{theorem}
\begin{proof}
	By the definition of $\lelog$, $L \lelog \nreg(F)$ iff there exists a $\logTM$ transducer~$T$ 
	that maps the input $x$ of the problem $x \stackrel{?}{\in} L$ to the input $T(x)$ of the problem 	$\nreg(F)$. We construct an $\NA_F\logTM$ $M$ recognizing $L$ via the composition of $T$ and $\NA_F\logTM$ $M_\nreg$ from Lemma~\ref{lemma:NregToNAFlogTM}. 
	
	So $\{ L \mid L \lelog \nreg(F) \} \subseteq \CL(\NA_F\logTM)$; the opposite inclusion follows from Lemma~\ref{lemma:NAFlogTMtoNreg}. We repeat the same arguments for the deterministic case and apply Lemma~\ref{lemma:DAFlogTMtoDreg} for the opposite inclusion.
\end{proof}

\subsection{\boldmath$\BP\logTM$ models}

In Section~\ref{sect:Cones} we exploited the following idea. In the case of a nondeterministic model,  performing queries one by one and proceeding the computation depending on the queries' results computationally equivalent to guessing all the queries results and verifying whether all the results were correct in the end (by testing whether obtained protocol was correct). In fact, this idea works even in the case of a deterministic model in Theorem~\ref{th:MP1DPBAtoDRR}. Now we exploit it again for $\logTM$-based models.

\begin{lemma}\label{lemma:AandBPlogTMequiv}
  \[
	\MP\NA_\PROT\logTM \sim_{\log} \MP\NBP\logTM\ \text{and}\
	\MP\DBP\logTM \lelog \MP\DA_\PROT\logTM.
  \]  
\end{lemma}
We provide only the proof idea since the proof follows our general technique
that we have applied above a lot. 

\begin{proof}[Proof idea]
Let $M_A$ be a $\NA_\PROT\logTM$, $M_B$ be a $\NBP\logTM$, and $x$ be an input word.
Since both kinds of $\logTM$s are nondeterministic, $M_A$ can guess and verify $M_B$'s successful run on $x$ provided that $M_B$'s protocol is 
written on the advice tape; $M_B$ can guess $y \in \PROT$ and a successful run of $M_A$ on $(x, y)$, and verify it: the transitions on configurations are simulated on log space and the fact $y \in \PROT$ is verified by performing subsequently the queries from the sequence~$y$.

The case of deterministic models is similar to Theorem~\ref{th:MP1DPBAtoDRR}. $M_A$ just simulates $M_B$ and tests whether the query words on the advice tape, queries an the results are the same as $M_B$ has during processing of the input. 
\end{proof}

Combining all together, we obtain the main theorem of the section.

\begin{theorem}\label{th:RRmain}
  \begin{align}
    \label{eq_:I}
    \nreg(\PROT) \sim_{\log} \CL(\NBP\logTM) = \{ L \mid L \lelog \nreg(\PROT) \},\\
    \label{eq_:IIa}
    \CL(\DBP\logTM) \lelog \reg(\PROT),\\
    \label{eq_:IIb}
			\CL(\DBP\logTM) \subseteq \{ L \mid L \lelog \reg(\PROT) \}.
  \end{align}
\end{theorem}

\begin{proof} We begin with the proof of~\eqref{eq_:I}. By Lemma~\ref{lemma:AandBPlogTMequiv}
	\begin{equation}\label{eq_:a}
		\MP\NBP\logTM \sim_{\log} \MP\NA_\PROT\logTM .
	\end{equation}
	By Lemmata~\ref{lemma:NAFlogTMtoNreg} and~\ref{lemma:NregToNAFlogTM}
	\begin{equation}\label{eq_:b}
		\CL(\NA_\PROT\logTM) \sim_{{\log}} \nreg(\PROT).
	\end{equation}
	Eqs.~\eqref{eq_:a} and~\eqref{eq_:b} imply
	\begin{equation}\label{eq_:Ia}
		\CL(\NBP\logTM) \sim_{\log} \nreg(\PROT) .
	\end{equation}
	By Theorem~\ref{th:CLsRRsummary} 
	\begin{equation}\label{eq_:c}
		\CL(\NA_\PROT\logTM) = \{ L \mid L \lelog \nreg(\PROT) \}.
	\end{equation}
	Eqs.~\eqref{eq_:a} and~\eqref{eq_:c} imply
	\begin{equation}\label{eq_:Ib}
			\CL(\NBP\logTM) = \{ L \mid L \lelog \nreg(\PROT) \}.
	\end{equation}
	Eqs.~\eqref{eq_:Ia} and \eqref{eq_:Ib} form \eqref{eq_:I}.

	Now we move to the proof of~\eqref{eq_:IIa} and~\eqref{eq_:IIb}. By Lemma~\ref{lemma:DAFlogTMtoDreg}
	\begin{equation}\label{eq_:cc}
		\CL(\DA_\PROT\logTM) \lelog \reg(\PROT).
	\end{equation}
	By Lemma~\ref{lemma:AandBPlogTMequiv}
	\begin{equation}\label{eq_:dd}
		\MP\DBP\logTM \lelog \MP\DA_\PROT\logTM.
	\end{equation}
	Eqs.~\eqref{eq_:cc} and~\eqref{eq_:dd} imply \eqref{eq_:IIa}.
	By Theorem~\ref{th:CLsRRsummary}
	\begin{equation}\label{eq_:ee}
			\CL(\DA_\PROT\logTM) = \{ L \mid L \lelog \reg(\PROT) \}.
	\end{equation}
	Eqs.~\eqref{eq_:dd} and~\eqref{eq_:ee} imply \eqref{eq_:IIb}.	
\end{proof}

\section{Applications}\label{sect:Applications}

In this section we prove the applications~(\ref{fact:P}-\ref{fact:NP}) described in Section~\ref{sec:results}.

\begin{theorem}
	Assertions (\ref{fact:P}-\ref{fact:NP}) hold.
\end{theorem}
\begin{proof}
	$\SAPROT$ was defined in Example~\ref{ex:SAPROT}. It was proved in \cite{RVCSR2018} that the problems $\ov{\mathrm{E\mm1NSA}} \sim_{\log} \nreg(\SAPROT)$ are $\PSPACE$-complete. So we obtain~\eqref{fact:PSPACEI} by applying Theorem~\ref{th:RRmain}. We prove~\eqref{fact:P} in the same way by combining the facts
$\DPROT \sim_{\mathrm{rat}} D_2$ (Example~\ref{ex:DPROT}) and $\nreg(D_2)$ is P-complete~\cite{RVRR2015DCFS}, 
and apply Lemma~\ref{lemma:leratRR} and Theorem~\ref{th:RRmain}.
	
To prove~(\ref{fact:PSPACEII}-\ref{fact:NP}) we use facts about the filters $\Per_k = \{ (w\#)^k \mid w \in \Sigma_k \}$, where $\Sigma_k$ is a $k$-letter alphabet. The problem $\nreg(\Per_1)$ is $\NP$-complete and
$\nreg(\Per_k)$, $k > 1$,
is $\PSPACE$-complete~\cite{ALRSS09,Vya09DMeng}. We construct set-protocols based on these languages as follows.
Let $\Gamma_{\wr} = \Sigma_k$, $\Gamma_{\query} = \{\ins, \test\}$, $\Gamma_\resp = \{+, -\}$.
	The response to the $\ins$-query is positive only for the first query, $\test$-queries are the same as in Example~\ref{ex:SAPROT}. We denote the language of correct protocols with $\Gamma_{\wr} = \Sigma_k$ as $\SPkPROT$. It is easy to see that $\Per_k \lerat \SPkPROT$: an FST $T$ maps words of the form $w\ins+w\test+\cdots w\test+$ to $w\#w\#\cdots w\#$ by replacing queries and responses by $\#$; the sequence of queries with responses 
	$\ins+,\test+,\ldots,\test+$ is verifiable via a finite state control (the inputs with invalid sequence are rejected by the FST), so $\nreg(\Per_k) \lelog \nreg(\SPkPROT)$ by Lemma~\ref{lemma:leratRR}.

        Now we prove that $\SPkPROT \lerat \Per_k$. The FST $T$ takes on the input a word $(w\#)^n$ and acts as follows.
While translating a block $w\#$ to the output,
        it has the following options: (i) change at least one letter, (ii) erase at least one letter and maybe change others,   (iii) add at least one letter and maybe change others, (iv) do not change~$w$. Until $T$ has not write $\ins$, it replaces $\#$ by $\test-$ in the cases (i-iii), and either by $\test-$ or by $\ins+$ in the case (iv). After $T$ wrote $\ins+$, it replaces $\#$ by $\test+$ in the case (iv) and either by $\test-$ or by $\ins-$ in the cases (i-iii). It is easy to see that $T((w\#)^n)$ consists of all correct protocols with either $w$ first $\ins$-query or without $\ins$-queries at all, and exactly $n$ queries. So $T(\Per_k) = \SPkPROT$ and assertions~(\ref{fact:PSPACEII}-\ref{fact:NP}) follows from Lemma~\ref{lemma:leratRR} and Theorem~\ref{th:RRmain}.
\end{proof}

\section{On computational complexity\\ of correct protocol languages}

Theorem~\ref{th:RReqNonEmp} essentially says that the computational complexity of the non-emptiness problem for ADS-automata is the same as the computational complexity of the NRR  problem for
the corresponding correct protocols languages.
It can be used to answer the question about the range of complexities of the non-emptiness problems 
for ADS-automata. It extends the known results about the complexity of RR
problems~\cite{VyaPIT13}. It appears that these complexities
are almost universal. It means that for any nonempty language $X$
there exists a language of correct protocols $\PROT$ such that $X$
is reducible to $\ov{\EP1\NBPA}$.  The reductions in the two directions 
differ. In one direction it is a log-space $m$-reduction. In the
other, we present the proof
only for Turing reductions in polynomial time.

\begin{theorem}\label{th:BP-universality}
For any nonempty language $X\subseteq \{0,1\}^*$ there exists a language of correct
protocols $\PROT$ such that
\[
X \lelog\ov{\EP1\NBPA}\; \stackrel{\text{\upshape Th.~\ref{th:RReqNonEmp}}}{\sim_{\log}} \; \nreg(\PROT) \leTp X.
\]
\end{theorem}

In the proof of Theorem~\ref{th:BP-universality} we use the language
of protocols defined as follows. 
Set $\Gamma_{\wr}=\{0,1\}$, $\Gamma_{\query} = \{\#,r\}$,
$\Gamma_{\resp} = \{+,-,r\}$. The relation $\valid $ is defined as
follows: $ \valid(\#)  =\{+,-\}$,  $\valid(r) = \{r\}$.
 The language of correct protocols $\PROT$
consists of protocols such that, for every query block $u_i\q_i\r_i$,
either $\q_i=\r_i = r$ and  $u_i=\eps$, or $\q_i = \#$, $\r_i = +$ and $u_i\in L$, or  $\q_i = \#$, $\r_i = -$ and $u_i\notin L$.
Here $L\subseteq \{0,1\}^*$ is a language depending on  $X$ in the statement of the theorem.

The exact choice of $L$ is complicated. So we start with the presentation  of basic ideas  behind the  proof of
Theorem~\ref{th:BP-universality}. We encode binary words using a log-space computable injective map $\sq\colon \{0,1\}^*\to\{0,1\}^*$ such that $\sq(X)\subseteq L$ and $\sq(\bar X)\subseteq \bar L$. It suffices for the first reduction in the theorem, $X \lelog \nreg(\PROT)$, since the protocol $\sq(x)\#+$ is correct iff $x\in X$.

For the second reduction, i.e., $ \nreg(\PROT) \leTp X$, we need much
more requirements.  Let $\A$ be an input automaton for $\nreg(\PROT)$
and $S$ be its state set. We are going to decide
$L(\A)\cap \PROT\ne \es$ in polynomial time using oracle calls of the oracle $X$. For this purpose we reduce the question $L(\A)\cap \PROT\ne \es$  to the question  $R\ne \es$ for some regular language $R\in\{y, n, \#, +, -, r\}^*$.
By definition, $w\in R$ if there exists an accepting run of $\A$ that processes a correct protocol $p$ such that $p$ is obtained from $w$ by substitutions of letters  $y$ and $n$ with words of $L$  and $\bar L$ respectively (different words may be used for different occurrences of the letters).
To check the correctness of the run processing the protocol, we need to compute, for all pairs $s',s''\in S$, all possible transitions from $s'$ to $s''$ by processing words from $L$ only as well as all possible transitions from $s'$ to $s''$ by processing words from $\bar L$ only.

Thus, the main part of the reduction consists of solving NRR problems $L(\A_{s's''})\cap L\ne \es$ and $L(\A_{s's''})\cap\bar L\ne \es$ for all pairs $s', s''\in S$. Here $\A_{s's''}$ are  auxiliary automata. The states and the transitions of $\A_{s's''}$ coincide with the states and the transitions of $\A$. The initial state of  $\A_{s's''}$ is $s'$ and the only accepting state is $s''$.

Note that $L(\A_{s's''})$ may be infinite and it causes the first difficulty: one need to consider arbitrary long words in the protocol language. To avoid this difficulty we require that any infinite regular language intersects both $L$ and $\bar L$. Therefore $L(\A_{s's''})\cap L\ne \es$ and $L(\A_{s's''})\cap\bar L\ne \es$ if $L(\A_{s's''})$ is infinite.

If $L(\A_{s's''})$ is finite, it means that the transition graph is DAG (after removing states that are not reachable  and coreachable in $\A_{s's''}$). The second difficulty: it might be exponentially many words in $L(\A_{s's''})$. Again, to overcome it, we pose specific requirements on $L$ to guarantee  that that verifying $L(\A_{s's''})\cap L\ne \es$ and  $L(\A_{s's''})\cap\bar L\ne \es$ requires polynomially many oracle calls of the oracle $X$.

Now we provide formal arguments for the above plan of proof.
We encode binary words by the injective map
\begin{equation}\label{eq:square}
\sq\colon  x \mapsto \beta(x)11\beta(x)11,
\end{equation}
where $\beta\colon\{0,1\}^*\to\{0,1\}^* $  is the morphism defined on the symbols as $\beta(0) =01$, $\beta (1) = 10$.

For an NFA $\A$ with the state set $S$ we define a relation
\begin{equation*}
		s' \run{u} s''
\end{equation*}
that holds if $\A$ can reach  $s''$ on processing $u$ starting from the state $s'$.

Now we list the requirements on the language $L$. 
\begin{enumerate}
\item\label{cond:sq} As it mentioned before, $\sq(X)\subseteq L$ and $\sq(\bar X)\subseteq \bar L$.
\item\label{cond:infinite} There exists a language $W\subseteq \{0,1\}^*$ such that both $W\cap L$ and $W\cap \bar L$ are recognized in polynomial time,  and,
for any NFA $\A$ over the alphabet $\{0,1\}$ and any pair of its states $s_1$, $s_2$, either $L(\A_{s_1s_2})$ is finite, or  there exist $w_1\in L\cap W$, $w_2\in \bar L\cap W$ such that $w_1\in L(\A_{s_1s_2})$ and $w_2\in L(\A_{s_1s_2})$. 
\item \label{cond:sparse} The language $W$ is sparse: $|W\cap\{0,1\}^{\leq n}|= \poly(n)$. Moreover, the lists of  words in $L\cap W\cap\{0,1\}^{\leq n}$ and, respectively, in $\bar L\cap W\cap\{0,1\}^{\leq n}$ can be  generated in polynomial time.
\item \label{cond:lex} If $|u|=|v|$, $u\ne v$,  and $uv\notin W$ then  $uv\in L$ iff $u\lex v$, where $\lex$ is the lexicographical order.
\item\label{cond:trash} If $|w|$ is odd and $w\notin W$ then $w\in L$. If $w = xx$ and $w\notin \sq (\{0,1\}^*)\cup W$ then  $w\in L$. 
\end{enumerate}

The sets of $L$-transitions and $\bar L$-transitions are defined as follows:
\[
\delta^L_{\A}(s) =\big \{s'\in S :\exists u\ s\run{u} s',\ u\in L\big\},\quad
\delta^{\bar L}_{\A}(s) =\big \{s'\in S :\exists u\ s\run{u} s',\ u\in \bar L\big\}.
\]

The main part of the proof of Theorem~\ref{th:BP-universality} is the following lemma.

\begin{lemma}\label{lm:Lcorrectness}
Let $L$ be a language satisfying  Requirements \ref{cond:sq}--\ref{cond:trash}.
Then there exists a polynomial time algorithm with the oracle $X$ that outputs the  sets $\delta^L_{\A}(s) $,
$\delta^{\bar L}_{\A}(s) $, where $\A$ is an input of $\nreg(\PROT)$ and $s$ is its state.
\end{lemma}

Before presenting the algorithm from  the lemma, we analyze  the most difficult case separately.

\begin{proposition}\label{prop:evencase}
Let $L$ be a language satisfying  Requirements \ref{cond:sq}--\ref{cond:trash} and   $\A$ be an NFA with the initial state $s_0$ and the unique accepting state $s_f$ such that $L(\A)$ is finite, $L(\A)\cap W = \es$, and each word in $L(\A)$ has an even length. Then conditions $s_f\in\delta^L_{\A}(s_0) $ and $s_f\in\delta^{\bar L}_{\A}(s_0)$ can be verified by  a polynomial time algorithm with the oracle~$X$. 
\end{proposition}

\begin{proof}
By solving the reachability problem, one can detect the set of reachable and coreachable states of $\A$. All other states can be deleted without affecting $L(\A)$. From now on, we assume that all the states $s\in S$ are reachable and coreachable. 

Since $L(\A)$ is finite, the transition graph of $\A$ is a DAG as well as all its subgraphs.
For each pair of states $s_1,s_2\in Q$, let $\ell(s_1,s_2)$ be the set
\[
\big\{ k: \exists u \ s_1\run{u}s_2\ \text{and}\ |u|=k\big\}.
\]
Using topological sorting, one can construct all the sets $\ell(s_1,s_2)$ in polynomial time by the backward induction based on the relation
\[
\ell(s_1, s_2) = \bigcup_{s\in N(s_1)} \big(1+\ell(s,s_2)\big),
\]
where $N(s_1)$ is the set of states that are reachable from $s$ in one move and $1+X = \{y: y = 1+x,\; x\in X\}$.

For  a positive integer  $\ell$  and a state  $s$  of $\A$, we define
\[
\Lft(s,\ell) = \big\{u : s_0\run{u} s, \ |u|=\ell\big\},\quad
\Rt(s,\ell) =\big\{u : s\run{u} s_f, \ |u|=\ell\big\}.
\]
We order the sets $\Lft(s,\ell)$ and $\Rt(s,\ell)$ in  the lexicographical order. 
Let $\min_0(s,\ell)$ be the minimal word in $\Lft(s,\ell)$, and $\max_0(s,\ell)$ be the maximal word  in $\Lft(s,\ell)$, and $\min_1(s,\ell)$  be the minimal word in  $\Rt(s,\ell)$, and   $\max_1(s,\ell)$  be the maximal word  in  $\Rt(s,\ell)$.

There exists  an inductive procedure that computes $\min_0(s,\ell)$, $\max_0(s,\ell)$, $\min_1(s,\ell)$, and  $\max_1(s,\ell)$ in polynomial time. The procedure also verifies the conditions $\Lft(s,\ell)\ne\es$, $\Rt(s,\ell)\ne \es$. We describe computation of  $\min_0(s,\ell)$, the other words are computed similarly.

Suppose that $u$ is the prefix of $\min_0(s,\ell)$ of the length $0\leq k<\ell$ (if $k=\ell$, then the procedure returns $u$ and stops). Let  $S_k =\{s': s_0\run{u} s'\} $. This set can be computed in polynomial time. If there exists $s'\in S_k$ and $s''\in S$ such that $s''\in\delta_{\A}(s',0)$ and $\ell-k-1\in \ell(s'',s)$, then $u0$ is a prefix of $\min_0(s,\ell)$ of the length $k+1$. Otherwise, if there exists $s'\in S_k$ and $s''\in S$ such that $s''\in\delta_{\A}(s',1)$ and $\ell-k-1\in \ell(s'',s)$, then $u1$ is a prefix of $\min_0(s,\ell)$. If both conditions are not satisfied, then  $\Lft(s)=\es$. 

According to Requirement~\ref{cond:lex} on $L$, 
if there exist a state $s$ and an integer $\ell$ such that $\min_0(s,\ell)\lex \max_1(s,\ell)$, then  $s_f\in\delta^L_{\A}(s_0) $. Otherwise, $\max_1(s,\ell)\lexeq \min_0(s,\ell)$ for all $s$, $\ell$. Since $L(\A)\cap W=\es$, in this case  $s_f\in\delta^L_{\A}(s_0)$ if and only if there exist a state $s$ and an integer $\ell$ such that $\min_0(s,\ell)= \max_1(s,\ell)$ and either $\min_0(s,\ell)=\beta(x)11 $ and $x\in X$ or $\min_0(s,\ell)\max_1(s,\ell)\notin \sq(\{0,1\}^*)$ due to Requirements~\ref{cond:sq}, \ref{cond:trash}. The  condition $x\in X$ can be verified by an  oracle call, the rest of conditions can be verified in polynomial time.

A similar check can be done for  the condition $s_f\in\delta^{\bar L}_{\A}(s_0) $.  It is equivalent to the following: there exist a state $s$ and an integer $\ell$ such that either $\min_1(s,\ell)\lex \max_0(s,\ell)$, or
 $\min_1(s,\ell)=\max_0(s,\ell)=\beta(x)11$  and $x\notin X$.
\end{proof}

\begin{proof}[Proof of Lemma~\ref{lm:Lcorrectness}]
The algorithm maintains the sets $S^+\subseteq S$, $S^-\subseteq S$. Initially, $S^+=S^-=\es$. We will prove that at the end $S^+=\delta^L_{\A}(s) $, $S^-  =\delta^{\bar L}_{\A}(s) $.
The algorithm analyzes states $s'\in S$ one by one and adds $s'$ to the sets $S^+$, $S^-$  according to the following rules.

In the first step, the algorithm decides whether $L(\A_{ss'})$ is infinite.  It can be done in polynomial time.  If the answer is `yes', then the algorithm adds $s'$ to both sets  $S^+ $, $S^- $  and continues with the next state. The correctness of this step is guaranteed by Requirement~\ref{cond:infinite}.

If the answer at the first step is `no', the lengths of words in $L(\A_{ss'})$  do not exceed $|S|$ (otherwise, there exists a run of $\A$ from $s$ to $s'$ containing a~cycle, which implies that  $L(\A_{ss'})$ is infinite).
In the second step, the algorithm checks whether $L(\A_{ss'})\cap L\cap W\ne \es$ and, respectively, whether $L(\A_{ss'})\cap \bar L\cap W\ne \es$. It can be done in polynomial time due to Requirement~\ref{cond:sparse}.  If the first condition holds, then the algorithm adds  $s'$ to $S^+$. If the second condition holds, then the algorithm  adds  $s'$ to $S^-$.

In the third step, the algorithm constructs an NFA $\A'$ recognizing  $L(\A_{ss'})\setminus W$. Let $P$ be the set of prefixes of all words in  $W\cap \{0,1\}^{\leq |S|}$.   Due to Requirement~\ref{cond:sparse},  $|P|= \poly (|S|)$ and $P$ can be constructed in polynomial time. 
The states of $\A'$ are  the pairs $(\tilde s,p)$, $\tilde s\in S$, $p\in P\cup\{\bot\}$.
The set of transitions $\delta_{\A'}((s,p), a)$ consists of pairs $(\tilde s,p')$ such that $\tilde s\in \delta_{\A}(s,a)$ and $p' = p a\in P$, and pairs $(\tilde s,\bot)$ such that  $\tilde s\in \delta_{\A}(s,a)$ and $p' = p a\notin P$. 
The set of transitions $\delta_{\A'}((s,\bot), a)$ consists of pairs
$(\tilde s,\bot)$ such that $\tilde s\in \delta_{\A}(s,a)$. The initial state is $(s,\eps)$. Accepting states are
pairs $(s', \bot)$ and $(s', p)$, where $p\notin W$. This definition implies that $\A'$ can be constructed in polynomial time. To prove the correctness of the
construction, note that processing a word $w\in W\cap\{0,1\}^{\leq |S|}$ from $(s,\eps)$ finishes at the state $(s', w)$ which is not accepting. For  a word $w\in L(\A_{ss'})\setminus W  $ there exists an accepting run of $\A_{ss'}$. The corresponding run of $\A'$ finishes at 
a state of the form $(s', w)$ or $(s',\bot)$. Thus, $w$  is accepted by $A'$.

In the fourth step, the algorithm
checks whether $L(\A_{ss'})\setminus W$ contains a~word of odd length. It can be done in polynomial time since words of odd length form a~regular language and the intersection of this language with  $L(\A_{ss'})\setminus W$ is recognized by an NFA with $2|S'|$ states, where $S'$ is the state set of $\A'$. If the answer is `yes',  then the algorithm adds  $q'$ to $S^+ $.  The correctness of this step is guaranteed by Requirement~\ref{cond:trash}.

In the fifth step, the algorithm constructs an NFA $\A''$ that accepts exactly the words of even length from $L(\A_{ss'})\setminus W$, apply to it the algorithm of Proposition~\ref{prop:evencase}, updates $S^+ $, $S^- $ if necessary, and continues with the next state.

It is  clear from the above remarks  that at the end  $S^+\subseteq\delta^L_{\A}(s)$, $S^-\subseteq\delta^{\bar L}_{\A}(s) $.

Suppose  that $s'\in\delta^L_{\A}(s) $. If $L(\A_{ss'})$ is infinite then $s'$ is added to $S^+$ at the first step. If  $L(\A_{ss'})\cap W\cap L\cap \{0,1\}^{\leq |Q|}\ne \es$, then $s'$ is added at the second step.
Otherwise, $(L(\A_{ss'})\cap L)\setminus W$ should be non-empty. If there are words of odd length in $L(\A_{ss'})\setminus W$, then $s'$ is added at the fourth step. And, finally, if $L(\A_{ss'})\setminus W$ consists of words of even length only, $s'$ is added at the fifth step due to Proposition~\ref{prop:evencase}. Therefore, $S^+=\delta^L_{\A}(s)$ at the end of the algorithm.

Suppose  that $s'\in\delta^{\bar L}_{\A}(s)$. If  $L(\A_{ss'})$  is infinite then $s'$ is added to $S^-$ at the first step. If  $L(\A_{ss'})\cap W\cap\bar L\cap \{0,1\}^{\leq |S|}\ne \es$, then $s'$ is added at the second step.
Otherwise, $(L(\A_{ss'})\cap \bar L)\setminus W\ne\es$ and  $s'$ is added at the fifth step due to Proposition~\ref{prop:evencase}. Therefore, $S^-=\delta^{\bar L}_{\A}(s)$ at the end of the algorithm.
\end{proof}

Now we prove that Requirements~\ref{cond:sq}--\ref{cond:trash} on $L$ are compatible.

\begin{lemma}\label{lm:Lchoice}
There exists $L$ satisfying   Requirements \ref{cond:sq}--\ref{cond:trash}. 
\end{lemma}
\begin{proof}
We define $W$ at first. For each triple $a,b,c$ of non-empty binary words there are two words in $W$ in the form $ab^{2r(a,b,c)}c$ and $ab^{2q(a,b,c)+1}c$ and for each $w\in W$ there exists a unique triple $a,b,c$ such that either $w = ab^{2r(a,b,c)}c$ or $w=ab^{2q(a,b,c)+1}c$. The definition of $W$ is inductive. Order all triples $x,y,z$  of non-empty binary words  with respect to  the length of $xyz$ and order the triples with the same length of $xyz$ with respect to the lexicographical order on the triples of binary words (binary words are also ordered lexicographically).

Assume that for all $(x,y,z)$ less than $(a,b,c)$ we have defined $s(x,y,z)$ and $t(x,y,z)$ properly.
Thus the set  $W'\subseteq W$ has been already defined. 
The total number of  $(x,y,z)$ less than $(a,b,c)$ does not exceed $\binom{|abc|-1}{2}\cdot(2^{|abc|}-1)$.
Thus, there are at most
$2\binom{|abc|-1}{2}\cdot(2^{|abc|}-1)$
 words from $W'$ having lengths in the range
 $[2^{3|abc|+3},2^{3|abc|+4}-1 ]$. So, there exist at least
\[
\frac{2^{3|abc|+3}}{2\binom{|abc|-1}{2}\cdot(2^{|abc|}-1)}-1 >2|abc| > 2|b|
\]
consecutive  integers $i$  in the range such that no word in  $W'$ has the length $i$. At least $|b|$ of them are even and at least $|b|$ of them are odd. It guarantees that the sets
\[
\begin{aligned}
&E= \{j: ab^{2j}c\notin W',\ 2^{3|abc|+3}\leq |ab^{2j}c|< 2^{3|abc|+4}\}
\ \text{and}\\ 
&O= \{j: ab^{2j+1}c\notin W',\ 2^{3|abc|+3}\leq |ab^{2j+1}c|< 2^{3|abc|+4}\}
\end{aligned}
\]
are non-empty.
 Set $r(a,b,c)$ be the minimal $j$ in $E$ and  $q(a,b,c)$ be the minimal $j$ in~$O$. 

To define $L$, we require that
the words from $W$ in the form $ab^{2q(a,b,c)+1}c$ are in $L$, while 
the words from $W$ in the form $ab^{2r(a,b,c)}c$ are in $\bar L$. Note
that it implies Requirement~\ref{cond:infinite}, since any infinite
regular language contains all words in the form $ab^kc$, $k>0$,  for some $a$, $b$, $c$.

By construction,  for each $w\in W $ the length of defining triple $a,b,c$ is logarithmic in the length of $w$. Thus $|W\cap\{0,1\}^{\leq n}|= \poly(n)$. To construct the list of words in $L\cap W\cap\{0,1\}^{\leq n}$ and the list of 
words in $\bar L\cap W\cap\{0,1\}^{\leq n}$ one need to perform only polynomial number of steps of the defining procedure and each step can be performed in polynomial time. Therefore, Requirement~\ref{cond:sparse} is satisfied.

 The rest of $L$ is defined to satisfy Requirements~\ref{cond:sq},~\ref{cond:lex}, and~\ref{cond:trash}. Note that $\sq(\{0,1\}^*)\cap W=\es$, since, for each $x\in\{0,1\}^*$,  $\sq(x)$ does not contain proper periodic subwords of length greater $|\sq(x)|/2$ but each word in $W$ do contain such words. It means that the construction of $W$ does not conflict with Requirement~\ref{cond:sq}.
\end{proof}

Now Theorem~\ref{th:BP-universality} follows from Lemma~\ref{lm:Lcorrectness} and Lemma~\ref{lm:Lchoice}.

\begin{proof}[Proof of Theorem~\ref{th:BP-universality}]
Choose $L$ as in the proof of Lemma~\ref{lm:Lchoice}. The reduction $X
\lelog \nreg(\PROT)$ is given by a map $x\mapsto \sq(x)\#+$. It is
clear that the map is computed in logarithmic space. The correctness of reduction follows from Requirement~\ref{cond:sq}.

Now we describe the second reduction, $ \nreg(\PROT) \leTp X$. Let $\A$ be an input NFA for $\nreg(\PROT)$. The reducing algorithm computes all sets $\delta^L_{\A}(s)$, $\delta^{\bar L}_{\A}(s)$ using Lemma~\ref{lm:Lcorrectness}.

Let $\B$ be an NFA with the same state set as $\A$. The alphabet of $B$ is $\{y, n, \#, +, -, r\}$. Transitions $\delta_{\B}(s, a)$ coincide with transitions  $\delta_{\A}(s, a)$ for $a\in\{\#, +, -, r\}$. For the rest of transitions,  $\delta_{\B}(s,y) =\delta^L_{\A}(s) $ and   $ \delta_{\B}(s,n) =\delta^{\bar L}_{\A}(s) $. The initial state and the accepting states of $\B$ and of $\A$ coincide.

Let $R= L(\B)\cap \big(y\#+ \mid n\#-\mid rr\big)^*$. Then $L(\A)\cap
\PROT\ne\es$ iff $R\ne \es$. The latter condition is verified in
polynomial time since $R$ is regular.
\end{proof}

\section*{Acknowledgments}

This work is supported by the Russian Science Foundation grant 20--11--20203.

\bibliographystyle{abbrv}
\bibliography{AuxDS-extended}

\end{document}